\documentclass{article}
\usepackage{spconf}
\usepackage{amsmath,amsthm,amssymb,amsfonts}
\usepackage{algorithmic}
\usepackage{graphicx}
\usepackage{textcomp}
\usepackage{xcolor}
\usepackage{mathtools}
\usepackage{algorithm}
\usepackage{cite}
\usepackage{booktabs}
\usepackage{subcaption}
\usepackage[utf8]{inputenc}

\newtheorem{definition}{Definition}

\newtheorem{theorem}{Theorem}

\newcommand\numberthis{\addtocounter{equation}{1}\tag{\theequation}}

\DeclareMathOperator*{\logdet}{log\,det}
\DeclareMathOperator*{\argmax}{arg\,max}

\DeclareMathOperator*{\maximize}{maximize\quad}

\renewcommand{\vec}[1]{\boldsymbol{\mathrm{#1}}}

\newcommand{\R}{\mathbb{R}}

\newcommand{\matx}[1]{\boldsymbol{\mathrm{#1}}}

\newcommand{\tr}[1]{\mathrm{tr}\left\{#1\right\}}

\begin{document}
\ninept

\title{Sampling and Reconstruction of Signals on Product Graphs}
\name{Guillermo Ortiz-Jiménez, Mario Coutino, Sundeep Prabhakar Chepuri, and Geert Leus\thanks{The code of this article can be found at https://gitlab.com/gortizji/product\_graphs.}}
\address{Delft University of Technology, The Netherlands}

\maketitle

\begin{abstract}
In this paper, we consider the problem of subsampling and reconstruction of signals that reside on the vertices of a product graph, such as sensor network time series, genomic signals, or product ratings in a social network. Specifically, we leverage the product structure of the underlying domain and sample nodes from the graph factors. The proposed scheme is particularly useful for processing signals on large-scale product graphs. The sampling sets are designed using a low-complexity greedy algorithm and can be proven to be near-optimal. To illustrate the developed theory, numerical experiments based on real datasets are provided for sampling 3D dynamic point clouds and for active learning in recommender systems.
\end{abstract}

\begin{keywords}
Active learning, graph signal processing, product graphs, sparse sampling, submodularity.
\end{keywords}

\section{Introduction}

Graph signal processing aims to extend basic concepts of classical signal processing developed for signals defined on Euclidean domains to signals that reside on irregular domains with a network structure \cite{gsp}. Oftentimes, large scale graphs appear as a product of several smaller graphs. For example, time series on a sensor network, can be factorized using a cycle graph to represent time, and a spatial map to represent the network \cite{jtv}. In genomics, the graph that relates the different phenotypes of a population is a product of the graphs that relate the different character phenotypes \cite{phenotype}. And, the movie ratings on a platform like Netflix can be viewed as signals living on the product of the social network of the users and the graph of relations among movies \cite{netflix}. More formally, we say that a graph is a product graph when its node set can be decomposed as a cartesian product of the nodes of two smaller factor graphs, and its edges are related with a known connection to the edges of the factors \cite{graph_theory}.

As the number of nodes increases within each factor, the total number of nodes in the product graph grows exponentially. When this happens, the tools that were designed to process data on small networks start to fail. In literature, this issue is commonly known as the \emph{curse of dimensionality}. To circumvent this issue, some authors have proposed exploiting the product structure of large graphs to parameterize graph signals with a reduced dimensionality \cite{moura}. 

In this paper, we focus on the reconstruction of graph signals that reside on the vertices of a product graph by just observing a small subset of its vertices. In particular, we propose using a structured sampling scheme with which we select a sparse subset of nodes from each factor, thereby observing the signal at a few specific nodes of the product graph. This approach contrasts with traditional graph signal processing methods which do not take into account any underlying graph factorization when designing the sampling set \cite{barbarossa,sampling_graph,sampling_graph_2,sampling_graph_3,covariance_journal,covariance_yonina,giannakis}. When the underlying graph factorization is not accounted for, the complexity of designing the sampling set scales with the total number of vertices in the graph, and therefore, the applicability of such methods to large graphs is very limited.

Our proposed scheme circumvents this issue by reducing the original product search space into the union of two much smaller spaces. Hence, avoiding the curse of dimensionality. An illustration of this is shown in Fig.~\ref{fig:sparsesensing}, where we see that selecting nodes from the factors reduces the possible candidate locations from 20 to 9. In essence, the aim is to reconstruct a signal on the product graph (rightmost in Fig.~\ref{fig:sparsesensing}) by observing a subset of nodes of the factor graphs (on the left of Fig.~\ref{fig:sparsesensing}) that generate the product graph. 

\begin{figure}[t!]
	\centering
	\includegraphics[width=\columnwidth]{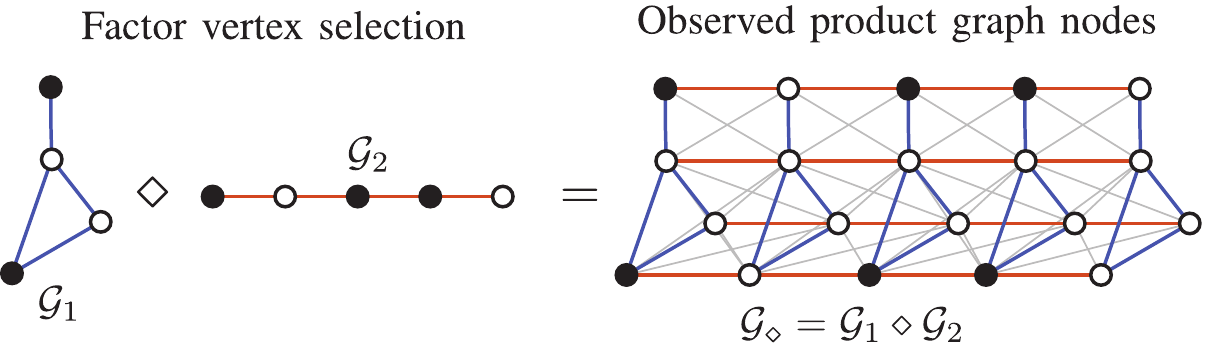}
	\caption{Illustration of the proposed sampling scheme for a product of two graphs. The black (white) dots represent the selected (unselected) vertices. $\diamond$ represents either a Cartesian (only colored edges), a Kronecker (only gray edges) or a strong product (all edges) between graphs. }
	\label{fig:sparsesensing}
\end{figure}

In this paper, we present a low-complexity near-optimal greedy algorithm to design such a structured subsampling scheme, and demonstrate its performance on real datasets related to dynamic 3D point clouds and recommender systems.

\section{Background and modeling}

Throughout this paper we will use upper (lower) case bold face letters to denote matrices (column vectors), and we will denote sets using calligraphic letters.

\subsection{Graph signals}

A graph signal $\matx{x}\in\R^{N}$ consists of a collection of $N$ values that can be associated to the nodes of a known graph $\mathcal{G}=(\mathcal{V},\mathcal{E})$, with a vertex set $\mathcal{V}$, and an edge set $\mathcal{E}$ that reveals the connections between the nodes. For the sake of exposition, we will focus on undirected graphs.

Using the graph structure, we can construct an adjacency matrix $\matx{A}\in\R^{N\times N}$ that stores the strength of the connection between edges $i$ and $j$ in its $(i,j)$th and $(j,i)$th entries, $[\matx{A}]_{i,j}=[\matx{A}]_{j,i}$. The degree of the $i$th node of a graph is defined as $d_i=\sum_{j=1}^N [\matx{A}]_{i,j}$. Related to the adjacency matrix, we can also define an alternative matrix representation known as the graph Laplacian $\matx{L}=\matx{D}-\matx{A}\in\mathbb{R}^{N\times N}$, where $\matx{D}=\mathrm{diag}\{d_1,\dots,d_N\}\in\R^{N\times N}$. Both the adjacency matrix and the graph Laplacian belong to the class of matrices that represent a valid \emph{graph-shift operator} \cite{gsp}, i.e., a matrix $\matx{S}\in\R^{N\times N}$ with a sparsity pattern defined by the graph connectivity. 

Since for undirected graphs $\matx{S}$ is symmetric, it admits an eigenvalue decomposition
\begin{equation}
	\matx{S}=\matx{U}\matx{\Lambda}\matx{U}^H=[\vec{u}_1\cdots\vec{u}_N]\mathrm{diag}\{\lambda_1,\dots,\lambda_N\}[\vec{u}_1\cdots\vec{u}_N]^H,\label{eq:evd}
\end{equation}
where the eigenvectors $\{\vec{u}_i\}_{i=1}^N$ and the eigenvalues $\{\lambda_i\}_{i=1}^R$ provide a notion of frequency for the graph setting \cite{gsp}. If working with directed graphs, one can simply replace \eqref{eq:evd} by a Jordan decomposition $\matx{S}=\matx{U}\matx{J}\matx{U}^{-1}$ \cite{jordan}.

The vectors $\{\vec{u}_i\}_{i=1}^N$ provide a Fourier-like basis for graph signals, allowing to decompose any signal $\vec{x}$ into its spectral components $\vec{x}_\mathrm{f}=\matx{U}^H\vec{x}$. In this sense, we say that a graph signal $\vec{x}$ is bandlimited \cite{gsp} when $\vec{x}_\mathrm{f}$ has $K<N$ non-zero entries.

\subsection{Product graphs}

Let $\mathcal{G}_1=(\mathcal{V}_1,\mathcal{E}_1)$ and $\mathcal{G}_2=(\mathcal{V}_2,\mathcal{E}_2)$ be two\footnote{For the sake of exposition, we restrict ourselves to the product of two graphs. Extensions to higher-order product graphs can be found in the journal version of this paper \cite{guille}.} undirected graphs with $N_1$ and $N_2$ vertices, respectively. The product graph \cite{graph_theory}, of $\mathcal{G}_1$ and $\mathcal{G}_2$ denoted with the symbol $\diamond$, is the graph given by
\begin{equation*}
	\mathcal{G}_\diamond=\mathcal{G}_1\diamond\mathcal{G}_2=(\mathcal{V}_1\times\mathcal{V}_2,\mathcal{E}_\diamond),
\end{equation*}
where $\mathcal{V}_\diamond=\mathcal{V}_1\times\mathcal{V}_2$ denotes the cartesian product of the vertex sets, and $\mathcal{E}_\diamond$ defines a valid edge set for $\mathcal{V}_\diamond$ according to the rules of the graph product. Depending on the set of rules that detemine $\mathcal{E}_\diamond$, three different product graphs are usually defined: the Cartesian product, the Kronecker product, and the strong product. For details of these rules we refer the reader to \cite{graph_theory}. The eigenvalue decomposition of the graph-shift operator of a product graph, denoted by $\matx{S}_\diamond$, is related to the eigenvalue decompositions of its factors through \cite{moura}
\begin{equation*}
	\matx{S}_\diamond=\matx{U}_\diamond\matx{\Lambda}_\diamond\matx{U}_\diamond^H=(\matx{U}_1\otimes\matx{U}_2)\matx{\Lambda}_\diamond(\matx{U}_1\otimes\matx{U}_2)^H,
\end{equation*}
where $\otimes$ denotes the Kronecker product between matrices, $\matx{U}_1$ and $\matx{U}_2$ are the eigenvectors of the graph-shift operators for $\mathcal{G}_1$ and $\mathcal{G}_2$, respectively, and $\matx{\Lambda}_\diamond$ is some diagonal matrix that depends on $\mathcal{G}_1$, $\mathcal{G}_2$ and the type of product.

Because every node in a product graph can be indexed using a pair of vertices of the factor graphs, we can rearrange any graph signal $\vec{x}\in\R^{N_1N_2}$ in a matrix form $\matx{X}\in\R^{N_2\times N_1}$ such that $\vec{x}_{i+(j-1)N_2}=[\matx{X}]_{i,j}$. The spectral decomposition of $\vec{x}$ then takes the form
\begin{equation}
	\vec{x}=(\matx{U}_1\otimes\matx{U}_2)\vec{x}_\mathrm{f}\Longleftrightarrow\matx{X}=\matx{U}_2\matx{X}_\mathrm{f}\matx{U}_1^T.\label{eq:spectral_2D}
\end{equation}

Using this formulation, we can say that a product graph signal $\matx{X}$ is bandlimited when it is simultaneously bandlimited in both domains, with a sparse $\matx{X}_\mathrm{f}\in\R^{N_2\times N_1}$ having $K_1<N_1$ columns and $K_2<N_2$ rows different than zero with known support. In such cases, ${\bf x}$ admits a low-dimensional representation as
\begin{equation}
	\vec{x}=(\tilde{\matx{U}}_1\otimes\tilde{\matx{U}}_2)\tilde{\vec{x}}_\mathrm{f}\Longleftrightarrow \matx{X}=\tilde{\matx{U}}_2\tilde{\matx{X}}_\mathrm{f}\tilde{\matx{U}}_1^T, \label{eq:reduced_model}
\end{equation}
where $\tilde{\matx{U}}_1$ and $\tilde{\matx{U}}_2$ are obtained by removing the columns of $\matx{U}_1$ and $\matx{U}_2$ corresponding to the indices of the rows and columns of $\matx{X}_\mathrm{f}$ that are zero, respectively; and $\tilde{\matx{X}}_\mathrm{f}$ and $\tilde{\vec{x}}_\mathrm{f}$ are the non-zero spectral components in $\matx{X}_\mathrm{f}$ and $\vec{x}_\mathrm{f}$. 

\section{Sampler design}

The low-dimensional representation (i.e., the notion of joint bandlimitedness) allows one to subsample graph signals by selecting just a few elements from each graph domain. Indeed, \eqref{eq:reduced_model} defines an overdetermined system of equations with $\tilde{K}=K_1K_2$ unknowns and $\tilde{N}=N_1N_2$ equations. Since
\begin{equation*}
	\mathrm{rank}(\tilde{\matx{U}}_1\otimes\tilde{\matx{U}}_2)=\mathrm{rank}(\tilde{\matx{U}}_1)\mathrm{rank}(\tilde{\matx{U}}_2),
\end{equation*}
we know that if $\tilde{\matx{U}}_1$ and $\tilde{\matx{U}}_2$ have full column rank we can recover $\tilde{\vec{x}}_\mathrm{f}$ from a subsampled version of $\vec{x}$.

Mathematically, sampling a subset of nodes from the graph factors is equivalent to selecting a subset of rows and columns from $\matx{X}$. Let $\mathcal{L}_1\subseteq\mathcal{V}_1$ and $\mathcal{L}_2\subseteq\mathcal{V}_2$ be two subsets of vertices from $\mathcal{G}_1$ and $\mathcal{G}_2$, respectively, with $|\mathcal{L}_1|=L_1\geq K_1$ and $|\mathcal{L}_2|=L_2\geq K_2$. To mathematically represent the proposed sampling scheme, we introduce two selection matrices $\matx{\Phi}_1(\mathcal{L}_1)\in\{0,1\}^{L_1\times N_1}$ and $\matx{\Phi}_2(\mathcal{L}_2)\in\{0,1\}^{L_2\times N_2}$. Then, the subsampled observations can be related to ${\bf X}$ using the following linear model
\begin{equation*}
	\matx{Y}=\matx{\Phi}_2(\mathcal{L}_2)\matx{X}\matx{\Phi}^T_1(\mathcal{L}_1)=\matx{\Phi}_2(\mathcal{L}_2)\tilde{\matx{U}}_2\tilde{\matx{X}}_\mathrm{f}\tilde{\matx{U}}_1^T\matx{\Phi}^T_1(\mathcal{L}_1),
\end{equation*}
which can equivalently be expressed in the vectorized form as
\begin{equation}
	\vec{y}=\left[\matx{\Phi}_1(\mathcal{L}_1)\tilde{\matx{U}}_1\otimes\matx{\Phi}_2(\mathcal{L}_2)\tilde{\matx{U}}_2\right]\tilde{\vec{x}}_\mathrm{f}. \label{eq:sampling}
\end{equation}

In the following, and in order to simplify the notation, whenever it will be clear, we will drop the explicit dependency of $\matx{\Phi}_1(\mathcal{L}_1)$ and $\matx{\Phi}_2(\mathcal{L}_2)$ on the sets of selected nodes $\mathcal{L}_1$ and $\mathcal{L}_2$, and we simply use $\matx{\Phi}_1$ and $\matx{\Phi}_2$.

One can estimate $\tilde{\vec{x}}_\mathrm{f}$ from $\vec{y}$ [cf. \eqref{eq:sampling}] using least-squares as
\begin{equation}
	\hat{\tilde{\vec{x}}}_\mathrm{f}=\left[\left(\matx{\Phi}_1\tilde{\matx{U}}_1\right)^\dagger\otimes\left(\matx{\Phi}_2\tilde{\matx{U}}_2\right)^\dagger\right]\vec{y},\label{eq:xf_hat}
\end{equation}
where $(\cdot)^\dagger$ denotes the Moore-Penrose left pseudo-inverse of a matrix, which due to the Kronecker structure can efficiently be computed separately for each domain. From $\hat{\tilde{\vec{x}}}_\mathrm{f}$ one can obtain $\hat{\vec{x}}$ using \eqref{eq:reduced_model}.

In the presence of additive white Gaussian noise, the performance of the least-squares solution depends on the proximity of the eigenvalues of the Fisher information matrix \cite{foundations}
\begin{align*}
	\matx{T}(\mathcal{L})&=\left(\matx{\Phi}_1\tilde{\matx{U}}_1\otimes\matx{\Phi}_2\tilde{\matx{U}}_2\right)^H\left(\matx{\Phi}_1\tilde{\matx{U}}_1\otimes\matx{\Phi}_2\tilde{\matx{U}}_2\right)\\
	&=(\matx{\Phi}_1\tilde{\matx{U}}_1)^H(\matx{\Phi}_1\tilde{\matx{U}}_1)\otimes(\matx{\Phi}_2\tilde{\matx{U}}_2)^H(\matx{\Phi}_2\tilde{\matx{U}}_2)\\
	&=\matx{T}_1(\mathcal{L}_1)\otimes\matx{T}_2(\mathcal{L}_2),\label{eq:kron_T}\numberthis
\end{align*}
where $\mathcal{L}=\mathcal{L}_1\cup\mathcal{L}_2$; and $\matx{T}_1(\mathcal{L}_1)$ and $\matx{T}_2(\mathcal{L}_2)$ are the Fisher information matrices of the sampled factor graphs. To design the sampling sets, we solve the following optimization problem
\begin{align*}
	\maximize_{\mathcal{L}_1,\mathcal{L}_2} & f\{\matx{T}(\mathcal{L})\}\numberthis \label{eq:sparsesensing}\\
	\text{subject to}\quad & |\mathcal{L}|=L\quad \mathcal{L}=\mathcal{L}_1\cup\mathcal{L}_2 \\
	& |\mathcal{L}_1|\geq K_1 \quad |\mathcal{L}_2|\geq K_2
\end{align*}
where possible candidates for $f\{\matx{T}\}$ are $\logdet\{\matx{T}\}$ \cite{cvx_sampling,logdet}, $\lambda_\text{min}\{\matx{T}\}$ \cite{correlated}, or $-\tr{\matx{T}^H\matx{T}}$ \cite{frame_potential}. In this paper, we use $\tr{\matx{T}^H\matx{T}}$, also known as frame potential, as a figure of merit, since it can be showed that minimizing the frame potential is directly related to minimizing the mean squared error of the reconstruction. Furthermore, we show that the use of the frame potential as objective function allows to develop a low-complexity algorithm that has a multiplicative near-optimality guarantee.

To prove that our algorithm achieves near-optimality we rely on the concept of \emph{submodularity} \cite{submodular_book}, a notion based on the property of diminishing returns that is useful for solving discrete combinatorial optimization problems of the form \eqref{eq:sparsesensing}. Submodularity can formally be defined as follows.

\begin{definition}[Submodular function \cite{submodular_book}]\label{def:submodular}
  A function $f:2^\mathcal{V}\rightarrow\R$ defined over the subsets of $\mathcal{V}$ is submodular if it satisfies that for every $\mathcal{X}\subseteq\mathcal{V}$, and $x,y\in\mathcal{V}\setminus\mathcal{X}$ we have
  \begin{equation}
    f(\mathcal{X}\cup\{x\})-f(\mathcal{X})\geq f(\mathcal{X}\cup\{x,y\})-f(\mathcal{X}\cup\{y\}). \label{eq:submodular_def_rep}
  \end{equation}
\end{definition}

If a submodular function is also monotone non-decreasing, i.e., $f(\mathcal{X})\leq f(\mathcal{Y})$, for all $\mathcal{X}\subseteq\mathcal{Y}$; and normalized, i.e., $f(\varnothing)=0$, then one can show that the solution of the greedy maximization Algorithm~\ref{alg:greedy_matroid}, $f(\mathcal{X}_\text{greedy})$, is 1/2 near-optimal with respect to the solution of $\max\{f(\mathcal{X}): \mathcal{X}\in\mathcal{I}\}$, where $\mathcal{M}=(\mathcal{V},\mathcal{I})$ is a matroid \cite{submodular_2}. We will show next that under some minor modifications \eqref{eq:sparsesensing} satisfies the aforementioned conditions.

\begin{algorithm}[t]
\begin{algorithmic}[1]
\REQUIRE{$\mathcal{X}=\varnothing$, $K$, $\mathcal{M}=(\mathcal{V},\mathcal{I})$}
\FOR{$i=1$ to $L$}
\STATE $s^*=\argmax_{s\notin\mathcal{X}}\{f(\mathcal{X}\cup\{s\}):\mathcal{X}\cup\{s\}\in\mathcal{I}\}$ 
\STATE $\mathcal{X}\leftarrow\mathcal{X}\cup\{s^*\}$
\ENDFOR
\RETURN $\mathcal{X}$
\end{algorithmic}
\caption{Greedy maximization of a submodular function subject to a matroid constraint}
\label{alg:greedy_matroid}
\end{algorithm}

Notice that we can express the frame potential of a sampled product graph as
\begin{equation*}
	F(\mathcal{L})\coloneqq\tr{\matx{T}^H\matx{T}}
	=\tr{\matx{T}_1^H\matx{T}_1\otimes\matx{T}_2^H\matx{T}_2}\coloneqq F_1(\mathcal{L}_1)F_2(\mathcal{L}_2),
\end{equation*}
where we see that this function can be factorized as a product of the frame potential of its factors. With a slight modification we can obtain a normalized, monotone non-decreasing, submodular function that can be employed as a surrogate cost function.

\begin{theorem}\label{thm:submodular_F}
	The set function $G:2^\mathcal{V}\rightarrow\R$ given by $G(\mathcal{S})=F(\mathcal{V})-F(\mathcal{V}\setminus\mathcal{S})$ is a monotone non-decreasing, normalized, and submodular function. Here, $\mathcal{V}\coloneqq\mathcal{V}_1\cup\mathcal{V}_2$,  and $\mathcal{S}\coloneqq\mathcal{S}_1\cup\mathcal{S}_2$, with $\mathcal{S}_1\subseteq\mathcal{V}_1$ and $\mathcal{S}_2\subseteq\mathcal{V}_2$.
\end{theorem}
\begin{proof}
	See Appendix.
\end{proof}

Under the change of variables $\mathcal{L}=\mathcal{V}\setminus\mathcal{S}$, maximizing $G(\mathcal{S})$ is the same as maximizing $-F(\mathcal{L})$, and thus $G(\mathcal{S})$ can be used as submodular cost function in \eqref{eq:sparsesensing}. Furthermore, the constraints in \eqref{eq:sparsesensing} form a truncated partition matroid \cite{discrete_optimization} $\mathcal{M}_t=(\mathcal{V},\mathcal{I}_u\cap\mathcal{I}_p)$ with $\mathcal{I}_u=\{\mathcal{S}\subseteq\mathcal{V}:|\mathcal{S}|\leq N-L\}$ and $\mathcal{I}_p=\{\mathcal{S}\subseteq\mathcal{V}:|\mathcal{S}\cap\mathcal{V}_1|\leq N_1-K_1, |\mathcal{S}\cap\mathcal{V}_2|\leq N_2-K_2\}$, where $N=N_1+N_2$. Therefore, the solution to the greedy algorithm that solves \eqref{eq:sparsesensing}, $G(\mathcal{S}_\text{greedy})$, is 1/2 near-optimal \cite{submodular_2}, i.e. $G(\mathcal{S}_\text{greedy})\geq \frac{1}{2}G(\mathcal{S}^*)$, where $\mathcal{S}^*$ is the optimal solution to \eqref{eq:sparsesensing}.

\section{Numerical results}

In this section, we test the performance of the proposed sampling scheme for two different applications: subsampling a dynamic point cloud of a  dancer, and active querying via subsampling for a recommender system.

\subsection{Dynamic 3D point cloud}

The moving dancer dataset consists of $N_1=573$ frames in which the 3D coordinates of $N_2=1502$ markers, placed on the body of a dancer, were recorded. To represent the data as a graph signal we build a spatial 5-nearest-neighbor graph with the time-averaged position of the $N_2$ markers as suggested in \cite{jtv}; and consider time as a cycle graph with $N_1$ vertices. The resulting product graph consists of more than $850,000$ nodes, and the dynamic signal can be represented using three matrices $\matx{X}_x, \matx{X}_y,\matx{X}_z\in\R^{N_2\times N_1}$.

A visual inspection of the spectral decomposition of these signals  shows that most of the energy is confined in the first few eigenmodes of the temporal and spatial graphs. In particular, in our simulations we limit the support of the signals to the first $K_1=500$ and $K_2=70$ spatial and temporal frequencies, respectively.

Using this representation, we run our greedy method with $L=600$ to select which nodes to sample. An illustration of the quality of the results for two different frames of the point cloud video is shown in Fig.~\ref{fig:dancer}. Even with this amount of compression on such a highly dimensional signal, the reconstruction is very accurate. The point cloud's shape is hardly distorted and the point deviations are very small. In addition, we compute the relative error between the original and the estimated graph signals and we obtain a relative error of 3.44\%. We emphasize here that the estimates are obtained by observing only 525 vertices from $\mathcal{V}_1$ and $75$ from $\mathcal{V}_2$, i.e., 4.57\% of all the vertices in the product graph. We also draw 1000 different random subsets of 600 vertices from $\mathcal{V}$ with at least $K_1$ and $K_2$ vertices selected from $\mathcal{V}_1$ and $\mathcal{V}_2$, respectively. However, all these random sampling sets lead to a singular system of equations, and hence the results are not shown here.

\begin{figure}
	\centering
	\includegraphics[width=\columnwidth]{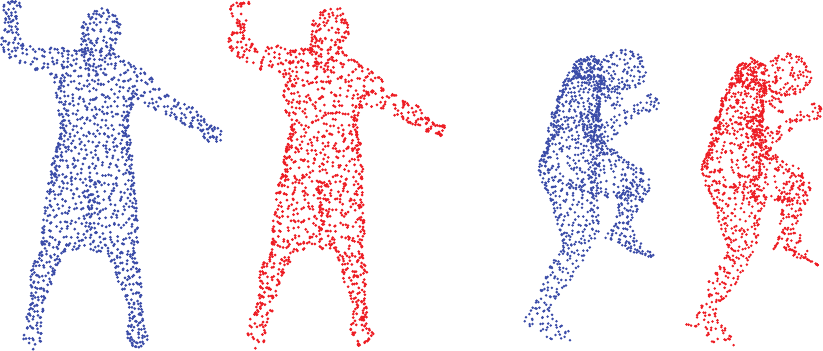}
	\caption{Two frames of the dancer dynamic point cloud. The blue dots correspond to the original data and the red dots to the subsampled version.}
	\label{fig:dancer}
\end{figure}

\subsection{Active learning for recommender system}

Most of the current recommendation algorithms solve the following estimation problem: given the past recorded preferences of a set of users for a set of products, what is the \emph{rating} that these users would give to some other set of products? In this paper, in contrast, we focus on the data acquisition phase of the recommender system. In particular, we claim that by carefully designing which users to poll and on which items, we can obtain an estimation performance on par with the state-of-the-art methods using only a fraction of the data that current methods require, and using a simple least-squares estimator. Thus, expensive and random querying can be avoided.

We showcase this idea on the MovieLens $100k$ dataset \cite{movielens} that contains partial ratings of $N_1=943$ users over $N_2=1682$ movies which are stored in a matrix $\matx{X}\in\R^{N_1\times N_2}$. This dataset also provides different features that can be used to build 10-nearest-neighbors graphs for the user and movie relations. This way, we can regard $\matx{X}$ as a signal living on the product of these two graphs.

The bandlimitidness of $\matx{X}\in\R^{N_1\times N_2}$ has already been exploited to impute its missing entries \cite{marques,kalofolias2014matrix}. In our experiments, we use $K_1=K_2=20$. Using this representation, we run our greedy algorithm with $L=100$, resulting in a selection of $L_1=25$ user and $L_2=75$ movie vertices, i.e., 1875 vertices in the product graph. We reconstruct our signals using \eqref{eq:xf_hat} and \eqref{eq:reduced_model} and compute the RMSE of the estimation using the test mask provided by the dataset. Nevertheless, since our active query method requires access to ground truth data which is not provided in the dataset, we use GRALS \cite{grals} to complete the matrix, and use its estimates when required.

\begin{figure}[t]
    \centering
    \begin{subfigure}[b]{0.46\columnwidth}
        \includegraphics[width=\textwidth]{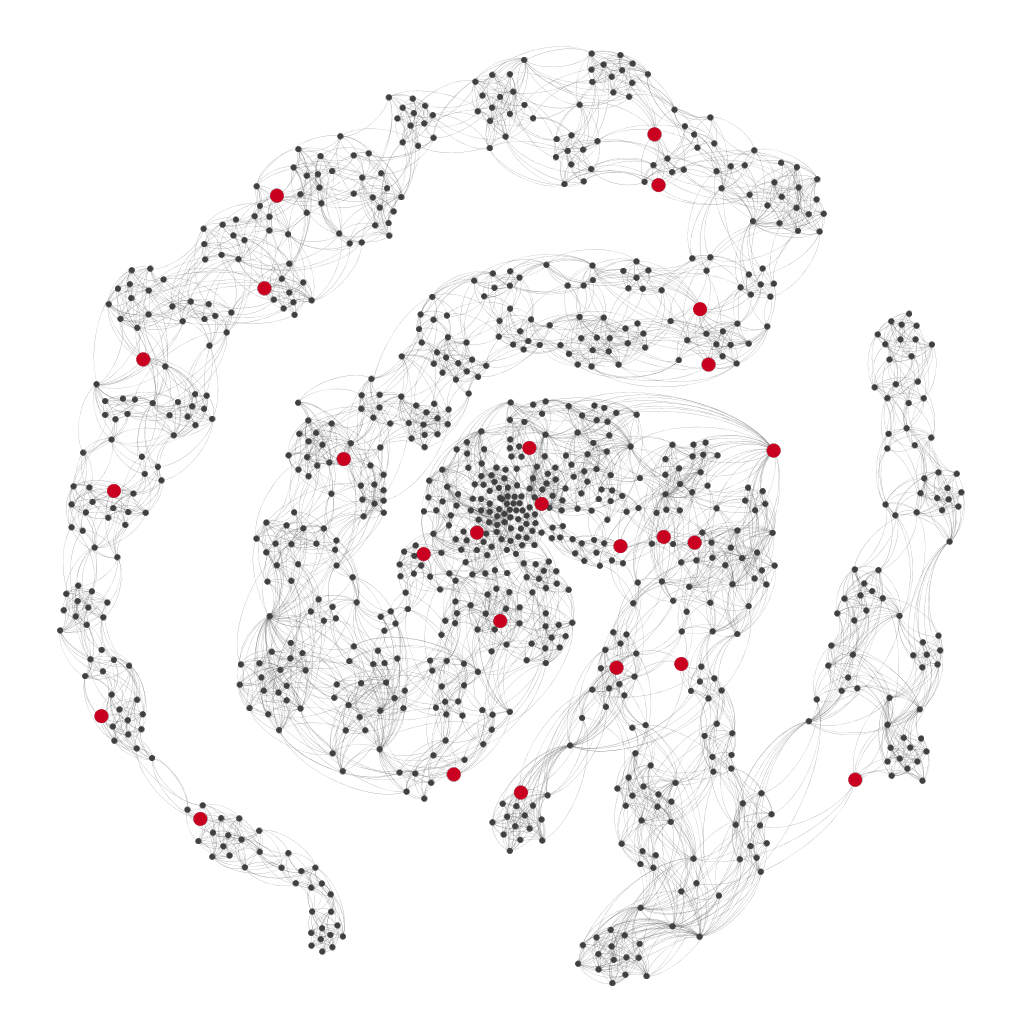}
        \caption{User graph}
        \label{fig:users}
    \end{subfigure}
    \hfill
    \begin{subfigure}[b]{0.46\columnwidth}
        \includegraphics[width=\textwidth]{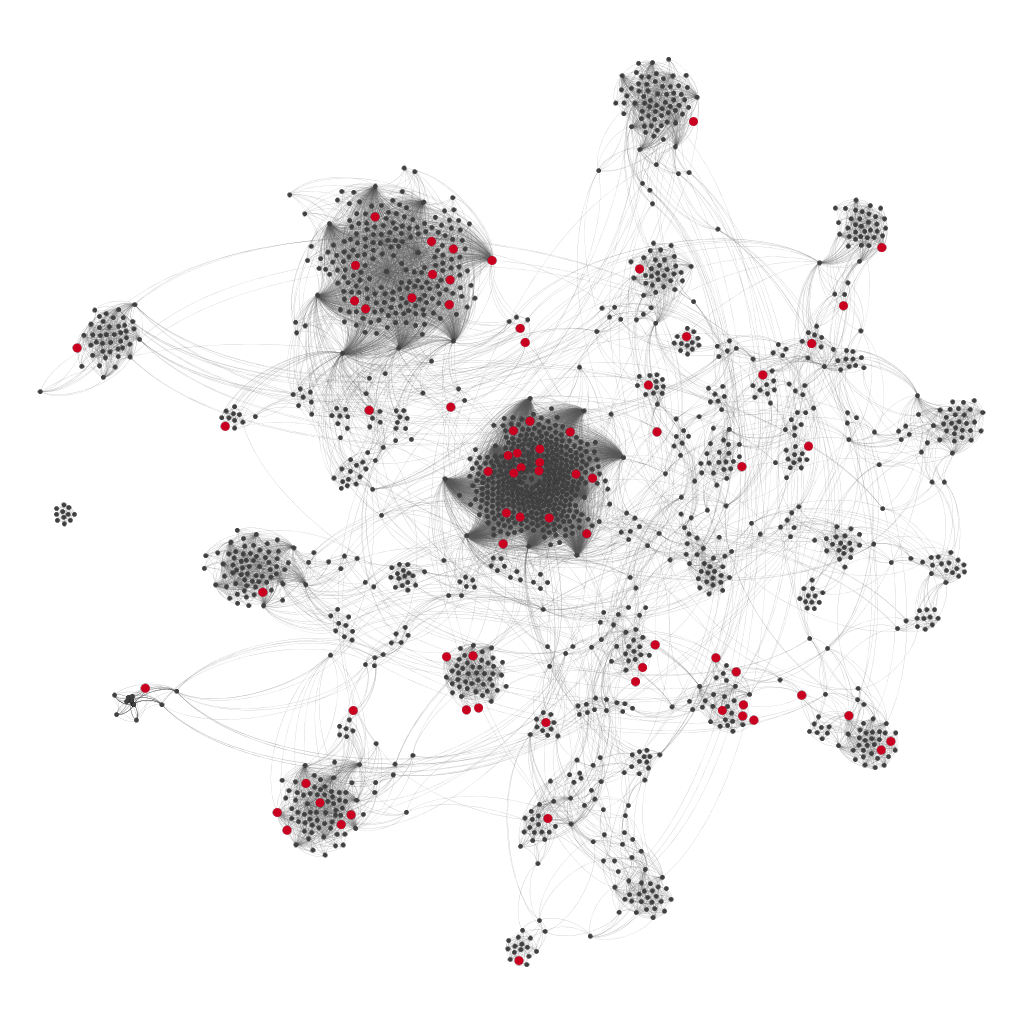}
        \caption{Movie graph}
        \label{fig:movies}
    \end{subfigure}
    \caption{User and movie graphs. The red (black) dots represent the observed (unobserved) vertices. Gephi visualization \cite{gephi}.}\label{fig:graphs}
\end{figure}

Fig.~\ref{fig:graphs}, shows the proposed active query sample resulting from the application of Algorithm~\ref{alg:greedy_matroid} to solve \eqref{eq:sparsesensing}. The user graph [cf. Fig.~\ref{fig:users}] is made out of small clusters connected in a chain-like structure, resulting in a uniformly spread distribution of observed vertices. On the other hand, the movie graph [cf. Fig.~\ref{fig:movies}] is made out of a few big and small clusters. Hence, the proposed active query sample assigns more observations to the bigger clusters and fewer observations to the smaller ones. We also compare the performance of the state-of-the-art methods to that of our algorithm [cf. Table~\ref{tab:performance}]. In light of these results, it is clear that a proper design of the sampling set allows to obtain the best performance with significantly fewer observed values, and using a much simpler non-iterative estimator.

\section{Conclusions}

In this paper, we have investigated the design of sparse samplers for the estimation of signals that reside on the vertices of a product graph. We have shown that by designing the sampling set using a combination of vertices from the graph factors we can overcome the curse of dimensionality, and design efficient subsampling schemes that guarantee a good performance for the reconstruction of graph signals. We have also proposed, a low-complexity greedy algorithm to select which vertices to sample, and provided a bound for its near-optimality with respect to the optimal subsampling set.

\begin{table}
\centering
\begin{tabular}{@{}ccc@{}}
\toprule
\textbf{Method} & \textbf{Number of samples} & \textbf{RMSE} \\ \midrule
GMC \cite{kalofolias2014matrix}            & 80,000                      & 0.996         \\ \midrule
GRALS  \cite{grals}         & 80,000                      & 0.945         \\ \midrule
sRGCNN\cite{monti2017geometric}         & 80,000                      & 0.929         \\ \midrule
GC-MC    \cite{berg2017graph}       & 80,000                      & \textbf{0.905}         \\ \midrule\midrule
Our method      & \textbf{1,875}                       & 0.9347        \\ \bottomrule
\end{tabular}
\caption{Performance on MovieLens $100k$. Baseline scores are taken from \cite{berg2017graph}.}
\label{tab:performance}
\end{table}
 
\appendix

\section{Proof of Theorem 1}
In order to simplify the derivations, let us introduce the notation
\begin{equation*}
  \bar{F}_i(\mathcal{S}_i)=F_i(\mathcal{V}_i\setminus\mathcal{S}_i)\quad i=1,2
\end{equation*}
so that $G(\mathcal{S})$ can also be written
\begin{equation*}
  G(\mathcal{S})\coloneqq F_1(\mathcal{V}_1)F_2(\mathcal{V}_2)-\bar{F}_1(\mathcal{S}_1)\bar{F}_2(\mathcal{S}_2).
\end{equation*}

We start by proving normalization, i.e.,
\begin{equation*}
	G(\varnothing)=F_1(\mathcal{V}_1)F_2(\mathcal{V}_2)-F_1(\mathcal{V}_1\setminus\varnothing)F_2(\mathcal{V}_2\setminus\varnothing)=0.
\end{equation*}

Now, recall that in \cite{frame_potential} they proved that the frame potential set function $F_i(\mathcal{L}_i)$ is a monotone non-decreasing supermodular function. The complementary function $\bar{F}_i(\mathcal{S}_i)=F_i(\mathcal{V}_i\setminus\mathcal{S}_i)$ preserves the supermodular property, but changes the monotonicity to non-increasing. Thus, when we invert its sign to obtain $-\bar{F}_i(\mathcal{S}_i)$ we obtain a monotone non-decreasing submdoular function. 

Furthermore, since the multiplication of two monotone non-decreasing functions results in a monotone non-decreasing function, and the addition of a constant preserves monotonicity, $G$ is a monotone non-decreasing set function.

Let $\{\mathcal{A}_1,\mathcal{A}_2\}$ be a partition of $\mathcal{S}$, with $\mathcal{A}_i\subseteq\mathcal{V}_i$ for $i=1,2$. To prove submodularity of $G(\mathcal{S})$ we use Definition~\ref{def:submodular}. However, since the ground set $\mathcal{V}_\diamond$ is now partitioned into the sets $\mathcal{V}_1$ and $\mathcal{V}_2$, there are two possible ways the elements $x$ and $y$ can be selected: either they both belong to the same set, or they belong to different sets. We prove that \eqref{eq:submodular_def_rep} is satisfied for both cases.

1. If $x,y\in\mathcal{V}_1$, then \eqref{eq:submodular_def_rep} can be developed as
  \begin{align*}
    \bar{F}_1&(\mathcal{A}_1)\bar{F}_2(\mathcal{A}_2)-\bar{F}_1(\mathcal{A}_1\cup\{x\})\bar{F}_2(\mathcal{A}_2)\\
    &\geq \bar{F}_1(\mathcal{A}_1\cup\{y\})\bar{F}_2(\mathcal{A}_2)- \bar{F}_1(\mathcal{A}_1\cup\{x,y\})\bar{F}_2(\mathcal{A}_2),
  \end{align*}
and simplifying
\begin{align*}
    \bar{F}_1(\mathcal{A}_1)-\bar{F}_1(\mathcal{A}_1\cup\{x\})\geq \bar{F}_1(\mathcal{A}_1\cup\{y\})- \bar{F}_1(\mathcal{A}_1\cup\{x,y\}).
\end{align*}

Multiplying both sides of the inequality by $-1$ we get
\begin{align*}
    \bar{F}_1(\mathcal{A}_1\cup\{x\})-\bar{F}_1(\mathcal{A}_1)\leq \bar{F}_1(\mathcal{A}_1\cup\{x,y\})-\bar{F}_1(\mathcal{A}_1\cup\{y\}),
\end{align*}
which is always satisfied since $\bar{F}_1$ is supermodular. The same proof holds if $x,y\in\mathcal{V}_2$.

2. If $x\in\mathcal{V}_1$ and $y\in\mathcal{V}_2$, then \eqref{eq:submodular_def_rep} can be developed as
  \begin{align*}
    \bar{F}_1&(\mathcal{A}_1)\bar{F}_2(\mathcal{A}_2)-\bar{F}_1(\mathcal{A}_1\cup\{x\})\bar{F}_2(\mathcal{A}_2)\\
    &\geq \bar{F}_1(\mathcal{A}_1)\bar{F}_2(\mathcal{A}_2\cup\{y\})- \bar{F}_1(\mathcal{A}_1\cup\{x\})\bar{F}_2(\mathcal{A}_2\cup\{y\}).
  \end{align*}
  Extracting the common factors, we obtain
  \begin{equation}
    \left[\bar{F}_1(\mathcal{A}_1)-\bar{F}_1(\mathcal{A}_1\cup\{x\})\right]\left[\bar{F}_2(\mathcal{A}_2)-\bar{F}_2(\mathcal{A}_2\cup\{y\})\right]\geq 0.\label{eq:cond_kron_3}
  \end{equation}
  Since $\bar{F}_1$ and $\bar{F}_2$ are non-increasing, we have 
  \begin{gather*}
    \bar{F}_1(\mathcal{A}_1)-\bar{F}_1(\mathcal{A}_1\cup\{x\})\geq 0\\
    \bar{F}_2(\mathcal{A}_2)-\bar{F}_2(\mathcal{A}_2\cup\{y\})\geq 0.
  \end{gather*}
  Thus, \eqref{eq:cond_kron_3} is always satisfied. The same procedure would hold for $x\in\mathcal{V}_2$ and $y\in\mathcal{V}_1$, thus proving that \eqref{eq:submodular_def_rep} is satisfied for any $\mathcal{S}\subseteq\mathcal{V}$ and $x,y\in\mathcal{V}\setminus\mathcal{S}$. Therefore, $G$ is submodular.
\newpage

\bibliographystyle{IEEEtran}
\bibliography{IEEEabrv,main}

\begin{thebibliography}{10}
\providecommand{\url}[1]{#1}
\csname url@samestyle\endcsname
\providecommand{\newblock}{\relax}
\providecommand{\bibinfo}[2]{#2}
\providecommand{\BIBentrySTDinterwordspacing}{\spaceskip=0pt\relax}
\providecommand{\BIBentryALTinterwordstretchfactor}{4}
\providecommand{\BIBentryALTinterwordspacing}{\spaceskip=\fontdimen2\font plus
\BIBentryALTinterwordstretchfactor\fontdimen3\font minus
  \fontdimen4\font\relax}
\providecommand{\BIBforeignlanguage}[2]{{%
\expandafter\ifx\csname l@#1\endcsname\relax
\typeout{** WARNING: IEEEtran.bst: No hyphenation pattern has been}%
\typeout{** loaded for the language `#1'. Using the pattern for}%
\typeout{** the default language instead.}%
\else
\language=\csname l@#1\endcsname
\fi
#2}}
\providecommand{\BIBdecl}{\relax}
\BIBdecl

\bibitem{gsp}
A.~Ortega, P.~Frossard, J.~Kova{\v c}evi{\'c}, J.~M.~F. Moura, and
  P.~Vandergheynst, ``{Graph Signal Processing: Overview, Challenges, and
  Applications},'' \emph{Proc. {IEEE}}, vol. 106, no.~5, pp. 808--828, May
  2018.

\bibitem{jtv}
F.~Grassi, A.~Loukas, N.~Perraudin, and B.~Ricaud, ``{A Time-Vertex Signal
  Processing Framework: Scalable Processing and Meaningful Representations for
  Time-Series on Graphs},'' \emph{{IEEE} Trans. Signal Process.}, vol.~66,
  no.~33, pp. 817--829, Feb 2018.

\bibitem{phenotype}
G.~P. Wagner and P.~F. Stadler, ``{Quasi-Independence, Homology and the Unity
  of Type: A Topological Theory of Characters},'' \emph{J. Theor. Biol.}, vol.
  220, no.~4, pp. 505 -- 527, Feb 2003.

\bibitem{netflix}
H.~Ma, D.~Zhou, C.~Liu, M.~R. Lyu, and I.~King, ``Recommender systems with
  social regularization,'' in \emph{Proc. ACM Int. Conf. Web Search and Data
  Mining}, 2011, pp. 287--296.

\bibitem{graph_theory}
W.~Imrich and S.~Klavzar, \emph{{Product graphs, structure and
  recognition}}.\hskip 1em plus 0.5em minus 0.4em\relax New York : Wiley, 2000.

\bibitem{moura}
A.~Sandryhaila and J.~M.~F. Moura, ``{Big Data Analysis with Signal Processing
  on Graphs: Representation and processing of massive data sets with irregular
  structure},'' \emph{{IEEE} Signal Process. Mag.}, vol.~31, no.~5, pp. 80--90,
  2014.

\bibitem{barbarossa}
M.~Tsitsvero, S.~Barbarossa, and P.~Di~Lorenzo, ``{Signals on graphs:
  Uncertainty principle and sampling},'' \emph{{IEEE} Trans. Signal Process.},
  vol.~64, no.~18, pp. 4845--4860, Sep.

\bibitem{sampling_graph}
X.~Zhu and M.~Rabbat, ``Approximating signals supported on graphs,'' in
  \emph{Proc. IEEE Int. Conf. Acoust., Speech, Signal Process.}, Kyoto, Japan,
  March 2012, pp. 3921--3924.

\bibitem{sampling_graph_2}
A.~Anis, A.~Gadde, and A.~Ortega, ``Towards a sampling theorem for signals on
  arbitrary graphs,'' in \emph{Proc. IEEE Int. Conf. Acoust., Speech, Signal
  Process.}, Florence, Italy, May 2014, pp. 3864--3868.

\bibitem{sampling_graph_3}
S.~Chen, R.~Varma, A.~Sandryhaila, and J.~Kovačević, ``{Discrete Signal
  Processing on Graphs: Sampling Theory},'' \emph{{IEEE} Trans. Signal
  Process.}, vol.~63, no.~24, pp. 6510--6523, Dec 2015.

\bibitem{covariance_journal}
S.~P. Chepuri and G.~Leus, ``{Graph Sampling for Covariance Estimation},''
  \emph{{IEEE} Trans. Signal Inf. Process. Netw.}, vol.~3, no.~3, pp. 451--466,
  Sept 2017.

\bibitem{covariance_yonina}
S.~Chepuri, Y.~Eldar, and G.~Leus, ``{Graph Sampling With and Without Input
  Priors},'' in \emph{Proc. IEEE Int. Conf. Acoust., Speech, Signal Process.},
  Calgari, Canada, April 2018, pp. 4564--4568.

\bibitem{giannakis}
D.~Romero, M.~Ma, and G.~B. Giannakis, ``Kernel-based reconstruction of graph
  signals,'' \emph{{IEEE} Trans. Signal Process.}, vol.~65, no.~3, pp.
  764--778, Feb.

\bibitem{jordan}
A.~Sandryhaila and J.~M.~F. Moura, ``Discrete signal processing on graphs:
  Frequency analysis,'' \emph{{IEEE} Trans. Signal Process.}, vol.~62, no.~12,
  pp. 3042--3054, Jun 2014.

\bibitem{guille}
G.~Ortiz-Jiménez, M.~Coutino, S.~P. Chepuri, and G.~Leus, ``Sparse sampling
  for inverse problems with tensors,'' \emph{arXiv preprint arXiv:1806.10976},
  2018.

\bibitem{foundations}
S.~P. Chepuri and G.~Leus, ``Sparse sensing for statistical inference,''
  \emph{Foundations and Trends in Signal Processing}, vol.~9, no. 3-4, pp.
  233--368, 2016.

\bibitem{cvx_sampling}
S.~Joshi and S.~Boyd, ``Sensor selection via convex optimization,''
  \emph{{IEEE} Trans. Signal Process.}, vol.~57, no.~2, pp. 451--462, Feb 2009.

\bibitem{logdet}
M.~Shamaiah, S.~Banerjee, and H.~Vikalo, ``{Greedy sensor selection: Leveraging
  submodularity},'' in \emph{Proc. 49th IEEE Conf. Decis. Control}, Atlanta,
  GA, USA, Dec 2010, pp. 2572--2577.

\bibitem{correlated}
S.~Rao, S.~P. Chepuri, and G.~Leus, ``Greedy sensor selection for non-linear
  models,'' in \emph{IEEE International Workshop on Computational Advances in
  Multi-Sensor Adaptive Processing}, Cancun, Mexico, Dec 2015, pp. 241--244.

\bibitem{frame_potential}
J.~Rainieri, A.~Chebira, and M.~Vetterli, ``{Near-Optimal Sensor Placement for
  Linear Inverse Problems},'' \emph{{IEEE} Trans. Signal Process.}, vol.~62,
  no.~5, pp. 1135--1146, March 2014.

\bibitem{submodular_book}
``Submodular functions and optimization,'' ser. Annals of Discrete Mathematics,
  S.~Fujishige, Ed.\hskip 1em plus 0.5em minus 0.4em\relax Elsevier, 2005,
  vol.~58.

\bibitem{submodular_2}
M.~L. Fisher, G.~L. Nemhauser, and L.~A. Wolsey, ``{An analysis of
  approximations for maximizing submodular set functions—II},'' in
  \emph{Polyhedral combinatorics}.\hskip 1em plus 0.5em minus 0.4em\relax
  Springer, Dec 1978, pp. 73--87.

\bibitem{discrete_optimization}
R.~G. Parker and R.~L. Rardin, ``{Polynomial Algorithms--Matroids},'' in
  \emph{Discrete Optimization}, ser. Computer Science and Scientific Computing,
  R.~G. Parker and R.~L. Rardin, Eds.\hskip 1em plus 0.5em minus 0.4em\relax
  San Diego: Academic Press, 1988, pp. 57--106.

\bibitem{movielens}
F.~M. Harper and J.~A. Konstan, ``The movielens datasets: History and
  context,'' \emph{ACM Trans. Interact. Intell. Syst.}, vol.~5, no.~4, p.~19,
  Jan 2016.

\bibitem{marques}
W.~Huang, A.~G. Marques, and A.~Ribeiro, ``Collaborative filtering via graph
  signal processing,'' in \emph{Proc. Eur. Signal Process. Conf.}, Kos, Greece,
  Aug 2017, pp. 1094--1098.

\bibitem{kalofolias2014matrix}
V.~Kalofolias, X.~Bresson, M.~Bronstein, and P.~Vandergheynst, ``Matrix
  completion on graphs,'' in \emph{Proc. Neural Inf. Process. Systems, Workshop
  "Out of the Box: Robustness in High Dimension"}, Montreal, Canada, Dec 2014.

\bibitem{grals}
N.~Rao, H.-F. Yu, P.~K. Ravikumar, and I.~S. Dhillon, ``{Collaborative
  filtering with graph information: Consistency and scalable methods},'' in
  \emph{Proc. Neural Inf. Process. Systems,}, Montreal, Canada, Dec 2015, pp.
  2107--2115.

\bibitem{gephi}
M.~Bastian, S.~Heymann, and M.~Jacomy, ``Gephi: An open source software for
  exploring and manipulating networks,'' in \emph{Int. AAAI Conf. Weblogs and
  Social Media}, San Jose, CA, USA, May 2009.

\bibitem{monti2017geometric}
F.~Monti, M.~Bronstein, and X.~Bresson, ``Geometric matrix completion with
  recurrent multi-graph neural networks,'' in \emph{Proc. Advances Neural Inf.
  Process. Systems}, Montreal, Canada, Dec 2017, pp. 3700--3710.

\bibitem{berg2017graph}
R.~v.~d. Berg, T.~N. Kipf, and M.~Welling, ``Graph convolutional matrix
  completion,'' \emph{arXiv:1706.02263}, 2017.

\end{thebibliography}

\end{document}